\documentclass[12pt]{article}
\newtheorem{proposition}{Proposition}

\newtheorem{corollary}{Corollary}

\newtheorem{example}{Example}
\newenvironment{proof}{\noindent{\bf Proof:}}{\hfill\fbox{}\vspace*{1mm}}
\usepackage[usenames]{color}
\usepackage{graphicx}

\begin{document}
\title{\bf On Pricing Basket  Credit Default Swaps }
\author{Jia-Wen Gu
\thanks{Advanced Modeling and Applied Computing Laboratory,
Department of Mathematics, The University of Hong Kong,
Pokfulam Road, Hong Kong. Email:jwgu.hku@gmail.com.
}
\and Wai-Ki Ching
\thanks{Advanced Modeling and Applied Computing Laboratory,
Department of Mathematics, The University of Hong Kong, Pokfulam
Road, Hong Kong. E-mail: wching@hku.hk. Research supported in
part by RGC Grants 7017/07P, HKU CRCG Grants
and HKU Strategic Research Theme Fund on Computational Physics and Numerical Methods.}
\and Tak-Kuen Siu
\thanks{ Department of Applied Finance and Actuarial Studies,
Faculty of Business and Economics, Macquarie University,
Macquarie University, Sydney, NSW 2109, Australia. Email: ken.siu@mq.edu.au,
ktksiu2005@gmail.com,}
\and Harry Zheng
\thanks{ Department of Mathematics,
Imperial College, London, SW7 2AZ, UK. Email: h.zheng@imperial.ac.uk.}
}
\date{}
\maketitle
\begin{abstract}
In this paper we propose a simple and efficient method to compute the ordered default time distributions in both the homogeneous case and  the two-group heterogeneous case under the interacting intensity default contagion model.
We give the analytical expressions for the ordered default time
distributions with recursive formulas for the coefficients, which makes the calculation fast and efficient in finding rates of basket CDSs.
In the homogeneous case, we { explore the ordered default time in limiting case and} further include the exponential decay and the multistate stochastic intensity process.
The numerical study indicates that, in the valuation of the swap rates and their sensitivities with respect to underlying parameters, our proposed model outperforms the Monte Carlo method.
\end{abstract}

\newpage
\section{Introduction}

Modeling portfolio default risk is a key topic in credit risk management.
It has important applications  and implications in pricing and hedging credit derivatives as well as risk measurement and
management of credit portfolios. There are two strands of literature
on credit risk analysis, namely, the structural firm value approach
pioneered by Black and Scholes (1973) and Merton (1974),
and the reduced-form intensity-based approach introduced by
Jarrow and Turnbull (1995) and Madan and Unal (1998).
In the classical firm value approach
the asset value of a firm is described by a geometric Brownian motion and
the default is triggered when the asset value  falls below a given default
barrier level. In the reduced-form intensity-based
approach, defaults are modeled as exogenous events
and their arrivals are described by using random point processes.

The reduced-form intensity-based approach has been
widely adopted for modeling portfolio default risk.
Two major types of reduced-form intensity-based models
for describing dependent defaults
are bottom-up models and top-down models.
Bottom-up models focus on modeling default intensities of
individual reference entities and their aggregation to form
a portfolio default intensity.
Some works on the bottom-up approach
for portfolio credit risk include Duffie and Garleanu (2001), Jarrow and Yu (2001),
Sch\"onbucher and Schubert (2001), Giesecke and Goldberg (2004),
Duffie, Saita and Wang (2006) and Yu (2007) etc.
These works differ mainly in specifying default intensities of
individual entities and their portfolio aggregation.
Top-down models concern modeling default at  portfolio
level. A default intensity for the whole portfolio
is modeled without reference to the identities of the
individual entities. Some procedures such as
random thinning can be used to recover the default
intensities of the individual entities.
Some works on top-down models
include Davis and Lo (2001), Giesecke and Goldberg (2005),
Brigo, Pallavicini and Torresetti (2006),
Longstaff and Rajan (2007) and Cont and Minca (2008).

One of the major applications of portfolio default risk
models is the valuation of credit derivatives written
on portfolios of reference entities.
Typical examples are collateralized debt obligations (CDOs)
and basket credit default swaps (CDSs).
The key to valuing these  derivatives is to know the portfolio
loss distribution function.
The $k${th} to default basket Credit Default Swap (CDS) is
a popular type of multi-name credit derivatives.
The protection buyer of a $k${th} to default basket CDS contract
pays periodic premiums to the protection seller of the contract
according to some pre-determined swap rates until the occurrence
of the $k${th} default in a reference pool. Whereas, the
protection seller of the $k${th} to default basket CDS
pays to the protection buyer of the contract the amount
of loss due to the $k${th} default in the pool
when it occurs. Different approaches have been
proposed in the literature for pricing the
$k${th} to default basket CDS under the intensity-based default contagion model. Herbertsson \& Rootz$\acute{e}$
(2006) introduce a matrix-analytic approach to value the $k$th
to default basket CDS. They transform the interacting intensity default
process to a Markov jump process which represents the default status in the portfolio.
This makes it possible to apply the matrix-analytic method to
derive a closed-form expression for the $k$th to default CDS.
Yu (2007) adopts the total hazard construction method by Norros (1987) and
Shaked \& Shanthikumar (1987) to generate default times with
a broad class of correlation structure.
He also compares this approach with the standard reduced-form models
and alternative approaches involving copulas.
Zheng \& Jiang (2009) use the total hazard construction to
derive the joint distribution of default times.
They give a closed-form expression for the joint distribution of
the general interacting intensity default process and
an analytical formula for valuing a basket CDS in a
homogeneous case.

Here we propose a simple and efficient method to
derive the $k$th default time distribution under the interacting intensity default contagion model.
We give the recursive formulas for the ordered default time
distributions, and further derive the analytic solutions
in a group of homogeneous entities and in two groups of heterogeneous entities.
In the homogeneous case, we {discuss the ordered default time in limiting case and} further include the exponential decay and the stochastic intensity process.
We derive the pricing formula under a two-state, Markovian regime-switching stochastic intensity model.
In addition, we show that our proposed method is superior to the simulation method in studying the sensitivities of the swap rates to changes of underlying parameters.

The rest of the paper is organized as follows.
Section 2 gives a snapshot for the interacting intensity-based default model.
Section 3 discusses the homogeneous case and applies the recursive method to characterize the ordered default time distributions derived in Zheng \& Jiang (2009).
Section 4 extends the method to study the multi-state stochastic intensity process.
Section 5 addresses the two-group heterogeneous case.
Section 6 presents numerical experiments for the evaluation of
the basket CDS under various scenarios and the sensitivity analysis.
Section 7 concludes the paper.

\section{A Snapshot for Interacting Intensity-Based Default Model}


In this section, we  give some preliminaries for the paper to facilitate our discussion.
Let $(\Omega, \mathcal{F}, \{\mathcal{F}_t\}_{t \geq 0}, P)$
be a complete filtered probability space, where we assume $P$ is a risk-neutral martingale measure (such a $P$ exists if we preclude the arbitrage opportunities), and
$\{\mathcal{F}_t\}_{t \geq 0}$ is a filtration satisfying the usual conditions,
(i.e., the right-continuity and $P$-completeness). We consider
a portfolio with $n$ credit entities. For each $i = 1, 2, \cdots, n$,
let $\tau_i$ be the default time of name $i$. Write
$N_i(t)=1_{\{\tau_i \leq t\}}$ for a single jump process
associated with the default time $\tau_i$, and
$\{ \mathcal{F}^i_t \}_{t \geq 0}$ is the right-continuous,
$P$-completed, natural filtration generated by $N_i$.
Suppose $\{ X_t \}_{t \geq 0}$ is the state process,
which represents the common factor process for joint
defaults. Write $\{ \mathcal{F}^X_t \}_{t \geq 0}$
for the right-continuous, $P$-completed, natural
filtration generated by the process $\{ X_t \}_{t \geq 0}$.
For each $t \geq 0$, write
$$
\mathcal{F}_t = \mathcal{F}_t^X \vee\mathcal{F}_t^1 \vee \ldots \vee \mathcal{F}_t^n \ .
$$
Here $\mathcal{F}_t$ represents the minimal $\sigma$-algebra containing information
about the processes $X$ and $\{ N_i \}^{n}_{i = 1}$ up to and including time $t$.

We assume that for each $i = 1, 2, \cdots,n$,
${N_i}$ possesses a nonnegative, $\{ \mathcal{F}_t \}_{t \geq 0}$-predictable,
intensity process $\lambda_i$ satisfying
$$
E\left( \int_0^t \lambda_i(s) ds \right) < \infty, \quad  t\geq 0,
$$
such that the compensated process:
$$
M_i(t) := N_i(t)-\int_0^{t \wedge \tau_i} \lambda_i(s) ds \ , \quad t \geq 0 \ ,
$$
is an $(\{\mathcal{F}_t\}_{t \geq 0}, P)$-martingale.
We further assume that the stochastic process $X$ is ``exogenous''
in the sense that conditional on the the whole path of $X$, (i.e.,
$\mathcal{F}_\infty^X$), $\lambda_i$ are $\{\bigvee_i {\mathcal{F}}_t^i\}_{t \geq 0}$-predictable.

To model the interacting intensity default process, we consider the following form:
\begin{equation}\label{e1}
\lambda_i(t) = a_{i}(t)+\sum_{j \neq i} b_{ij}(t) e^{-d_{ij}(t-\tau_j)}1_{\{\tau_j \leq t\}},
\quad i=1,2,\ldots,n,
\end{equation}
where $a_{i}(t)$ and $b_{ij}(t)$ are $\mathcal{F}^X$-adapted processes, and $d_{ij}$ are positive constants representing the rates of decay.
The introduction of exponential decay into the intensity-based default process is of practical significance, which indicates once a default occur in the portfolio, its effect on the other surviving entities will decrease at a rate proportional to its present impact.
Intensity rate processes in (\ref{e1}) determine the probability laws of the default times.
To price the $k$th to default basket CDS,
the distribution of the $k$th default time has to be known.

\section{Homogeneous Case}

In this section, we present a simple method to derive
the distribution of the $k$th \ default time in a group of
homogeneous entities under the interacting intensity default model.
Our method is based on the $k$th default rate and
the distribution of the random duration between
two defaults, where the contagion intensity process follows:
\begin{equation}\label{p1}
\lambda_i(t) = a \left(1  +  \sum_{j\neq i} c e^{-d(t-\tau_j)} 1_{\{\tau_j \leq t\}} \right),
\quad i = 1, 2, \ldots, n \ ,
\end{equation}
where $a$ is a positive constant and $c$ and $d$ are nonnegative constants.

We note that the Markov chain approach (Herbertsson \& Rootz$\acute{e}$(2006)) cannot solve this kind of processes with exponential decay ($d>0$).
Zheng \& Jiang(2009) adopt the total hazard construction method to give the joint distribution of default time $\tau_k$, $k=1,2,\ldots, n$,
while finding the ordered default time distributions remains to be intractable.
Here we give the recursive formula of the joint distribution of the ordered default times by our proposed method.
Let $\lambda^{k+1}(t)$ be the $(k+1)$th default rate at
time $t$ that triggers $\tau^{k+1}$.
Then $\lambda^{k+1}(t)$ will be the sum of the default rates of the surviving
entities after $\tau^k$. Under the homogeneous situation,
given the realization of $\tau^i, i=1,\ldots,k$, the ($k+1$)th default rate is:
$$
\lambda^{k+1}(t)= a(1+\sum_{i=1}^k c e^{-d(t-\tau^i)})(n-k),
$$
where $\tau^k < t \leq \tau^{k+1}$.
Then we have
$$
P(\tau^{k+1}-\tau^k > t \mid \tau^i, i=1,\ldots,k)= \exp\{-\int_0^t \lambda^{k+1}(\tau^k+s)ds\}
$$
Thus,
$$
P(\tau^{k+1}>t \mid \tau^i, i=1,\ldots,k)=\exp\{-\int_{\tau^k}^t \lambda^{k+1}(s)ds\},
$$
where $t \geq \tau^k$, which implies
\begin{equation}\label{a1}
f_{\tau^{k+1} \mid \tau^i, i=1,\ldots,k}(t)= a(1+\sum_{i=1}^k c e^{-d(t-\tau^i)})(n-k)\exp\{-\int_{\tau^k}^t a(1+\sum_{i=1}^k c e^{-d(s-\tau^i)})(n-k)ds\},
\end{equation}
where $f_{\tau^1}(t)=nae^{-nat}.$
One can apply (\ref{a1}) to derive the joint density function of $\tau^1, \tau^2, \ldots, \tau^{k+1}$ with the following recursion:
$$
f_{\tau^1, \tau^2, \ldots,\tau^{k+1}}(t_1,t_2, \ldots, t_{k+1})=f_{\tau^{k+1} \mid \tau^i =t_i, i=1,\ldots,k}(t_{k+1})f_{\tau^1, \tau^2, \ldots,\tau^{k}}(t_1,t_2, \ldots, t_{k}) \ ,
$$
where $t_1 < t_2 < \ldots < t_{k+1}$.
The unconditional density function of $\tau^k$ is given by the integral:
\begin{equation}
f_{\tau^k}(t)=
\displaystyle \int_0^{t} \int_0^{t_{k-1}} \cdots \int_0^{t_2}   f_{\tau^1, \tau^2, \ldots,\tau^{k-1},\tau^{k}}(t_1,t_2, \ldots,t_{k-1}, t)     dt_1 \cdots dt_{k-2} dt_{k-1}\\
\end{equation}
\begin{example}
If $n=2$, then the joint density function of $\tau^1$ and $\tau^2$ is:
$$
f_{\tau^1,\tau^2}(t_1,t_2)=
\left\{
\begin{array}{ll}
2a^2(1+ce^{-d(t_2-t_1)})\exp\{-a(t_1+t_2)+\frac{ac}{d}(e^{-d(t_2-t_1)}-1)\}, & t_1 \leq t_2\\
0, &t_1>t_2.
\end{array}
\right.
$$
The unconditional density function of $\tau^1$ is given by
$$
f_{\tau^1}(t)= 2a e^{-2at},
$$
and that of $\tau^2$ is given by
$$
f_{\tau^2}(t)= 2a^2 \int_0^t(1+ce^{-d(t-t_1)})\exp\{-a(t+t_1)+\frac{ac}{d}(e^{-d(t-t_1)}-1)\}dt_1.
$$
\end{example}
If we simplify our model by assuming that $d=0$, we have the following proposition.
\begin{proposition}
Suppose there are $n$ entities in our portfolio,
where the contagion intensity  process follows
\begin{equation}\label{inten1}
\lambda_i(t) = a \left(1  +  \sum_{j\neq i} c  1_{\{\tau_j \leq t\}} \right),
\quad i = 1, 2, \ldots, n \ .
\end{equation}
{ The $k$th default time $\tau^k$ is the sum of $k$ independent exponential random variables, i.e.,
$$
\tau^k=\sum_{i=0}^{k-1}X_i,
$$
where $X_i$ are independent exponential random variables with rates $a(1+ic)(n-i)$
respectively.
}
The the unconditional density function of $\tau^k$, k=1,2,\ldots,n,
satisfy the following recursive formula:
\begin{eqnarray}\label{p2}
\begin{array}{lll}
f_{\tau^{k+1}}(t) &=&
\displaystyle a(1+kc)(n-k) \int_0^t f_{\tau^k}(u) e^{-a(1+kc)(n-k)(t-u)} du,
\end{array}
\end{eqnarray}
where the initial condition is:
$$
f_{\tau^1}(t)=nae^{-nat}.
$$
\end{proposition}
\begin{proof}
We note that $\tau^{k+1}-\tau^{k}$ is independent of { $\tau^i$ for $i=0, 1,\ldots, k$},
where $\tau^0$ is assigned to be $0$.
Due to the homogeneous and symmetric properties of the entities in our portfolio,
we have
$$
\lambda^{k+1}(t) = a(1+kc)(n-k) \quad {\rm for} \quad \tau^k < t \leq \tau^{k+1}.
$$
Therefore
$$
P(\tau^{k+1}-\tau^k > t) = e^{-a(1+kc)(n-k)t}
$$
and this implies
$$
f_{\tau^{k+1}-\tau^{k}}(t)
= a(1+kc)(n-k)e^{-a(1+kc)(n-k)t}.
$$
{Let
$$X_i=\tau^{i+1}-\tau^i,$$
then
$$\tau^k=\sum_{i=0}^{k-1}X_i,$$
where $X_i$ are independent exponential random variables with rates $a(1+ic)(n-i)$.}
By convolution,
$$
\begin{array}{lll}
f_{\tau^{k+1}}(t) &=& \displaystyle  \int_0^t f_{\tau^k}(u)f_{\tau^{k+1}-\tau^{k}}(t-u)du\\
&=& \displaystyle a(1+kc)(n-k) \int_0^t f_{\tau^k}(u) e^{-a(1+kc)(n-k)(t-u)} du,
\end{array}
$$
for $k=1,2,\ldots,n-1$, where
$$
f_{\tau^1}(t)= f_{\tau^{1}-\tau^{0}}(t)=nae^{-nat}.
$$
\end{proof}

\begin{corollary}
Assume that $c \neq 1/i$ for $i=1,2,\ldots,n-1$. Then the
unconditional density function of $\tau^{k}$ is given by:
\begin{equation}
f_{\tau^k}(t) = \sum_{j=0}^{k-1}\alpha_{k,j}a e^{-\beta_ja t},
\end{equation}
where the coefficients are given by:
$$
\left\{
\begin{array}{l}
\alpha_{k+1,j}= \left\{
\begin{array}{lll}
\displaystyle \frac{\alpha_{k,j} \beta_k}{\beta_k-\beta_j},& j=0,1,\ldots,k-1\\
\displaystyle -\sum_{u=0}^{k-1}\frac{\alpha_{k,u} \beta_k}{\beta_k-\beta_u},& j=k
\end{array}
\right.\\
\displaystyle \beta_j = (n-j)(1+jc)
\end{array}
\right.
$$
and $\alpha_{1,0}=n$.
\end{corollary}
Applying the recursive formula (4) iteratively gives
the above corollary
in which the analytic expression are given with recursive formulas for the coefficients.
From the view point of computational convenience, we can see the advantage of the recursive formulas for the coefficients.
The same result is also obtained in Zheng \& Jiang (2009), with
$$
\alpha_{k,j}=\frac{(-1)^{k-1-j} n! (\prod_{m=1}^{k-1}(1+mc))}{(n-k)! j! (k-1-j)! (\prod_{m=0,m\neq j}^{k-1}(1+(m+j-n)c))}.
$$

We remark that the method to derive the recursive formula of $\tau^k$ stated here
is related to the total hazard construction method adopted
by Yu (2007) and Zheng \& Jiang (2009).
Assume we first enter the market right after the $k$th default
of the $n$ entities. Then $\tau^{k+1}-\tau^k$ is the first default time
being observed.
We draw a collection of independent standard exponential random variables
$E_1, E_2,\ldots, E_{n-k}$.
By using the total hazard construction method, we have
$$
\tau^{k+1}-\tau^k = \min \left\{\frac{E_1}{a(1+kc)},\frac{E_2}{a(1+kc)},\ldots,\frac{E_{n-k}}{a(1+kc)}\right\},
$$
which implies that
$$
P(\tau^{k+1}-\tau^k > t) = \prod_{i=1}^{n-k}P(E_i > a(1+kc)t) = e^{-a(1+kc)(n-k)t} \ .
$$
{ We note that the computational cost for 
the density of $\tau^k$ can grow up quickly, when $k$ is getting large.
The following propositions demonstrate the behavior of $\tau^k$ when $k$ and $n$ are large.
In Propositions 2 and 3, 
we temporarily define $\tau^k(n)$ as $k$th default time in a portfolio of $n$ names.
\begin{proposition}
For a fixed $k$, when $n \rightarrow \infty$, $\tau^k(n) \rightarrow 0$ almost surely.
\end{proposition}
\begin{proof}
For any given $\epsilon > 0$,
by Markov's inequality,
$$
P(\tau^k(n) \geq \epsilon) \leq \frac{E[\tau^k(n)]^2}{\epsilon^2}
$$
From Proposition 1,
$$
E[\tau^k(n)]=\sum_{i=0}^{k-1}E[X_i]=\sum_{i=0}^{k-1}\frac{1}{a(1+ic)(n-i)}<\frac{k}{a(n-k)},
$$
$$
Var[\tau^k(n)]=\sum_{i=0}^{k-1}Var[X_i]=\sum_{i=0}^{k-1}\frac{1}{(a(1+ic)(n-i))^2}<\frac{k}{a^2(n-k)^2}.
$$
Since
$$
\sum_{n=k+1}^{\infty} E[\tau^k(n)^2]<\sum_{n=k+1}^{\infty} \frac{k^2}{a^2(n-k)^2}+\sum_{n=k+1}^{\infty} \frac{k}{a^2(n-k)^2} < \infty,
$$
then from Borel-Cantelli Lemma, we have with probability 1, 
for all large $n$, $\tau^k(n) < \epsilon$. 
Hence $\tau^k(n) \rightarrow 0$ almost surely.
\end{proof}
}
{
\begin{proposition}
Let $k \rightarrow \infty$(due to $n > k$, in this case, $n \rightarrow \infty$), $\tau^k(n) \rightarrow 0$ almost surely and in particular,
$\tau^n(n) \rightarrow 0$ almost surely as $n \rightarrow \infty$.
\end{proposition}
\begin{proof}
From Proposition 1,
$$
E[\tau^k(n)]=\sum_{i=0}^{k-1}E[X_i]=\sum_{i=0}^{k-1}\frac{1}{a(1+ic)(n-i)}<\sum_{i=0}^{k-1}\frac{1}{a(1+ic)(k-i)}<\frac{2c+2}{a(ck+1)}(1+\frac{\ln (1+ck)}{c}).
$$
Indeed,
$$
\begin{array}{lll}
\displaystyle \sum_{i=0}^{k-1}\frac{1}{a(1+ic)(k-i)} &=& \displaystyle \sum_{i=0}^{k-1}\frac{1}{a(ck+1)}[\frac{1}{k-i}+\frac{c}{1+ci}]\\
&=& \displaystyle \sum_{i=0}^{k-1}\frac{1}{a(ck+1)}[\frac{1}{i+1}+\frac{c}{1+ci}] \\
 &<& \displaystyle \frac{2c+2}{a(ck+1)}\sum_{i=0}^{k-1}\frac{1}{1+ci}\\
&<&\displaystyle \frac{2c+2}{a(ck+1)}(1+\int_0^k \frac{1}{1+cx} dx)\\
&=&\displaystyle \frac{2c+2}{a(ck+1)}(1+\frac{\ln ( 1+ck)}{c}).
\end{array}
$$
Similarly,
$$
Var[\tau^k(n)]=\sum_{i=0}^{k-1}Var[X_i]<\sum_{i=0}^{k-1}\frac{1}{(a(1+ic)(k-i))^2}< \frac{(2c+2)^2}{a^2(ck+1)^2}\sum_{i=0}^{\infty}\frac{1}{(1+ic)^2}.
$$
Since
$$
\sum_{k=1}^{\infty}(E[\tau^k(n)])^2<\sum_{k=1}^{\infty}\frac{(2c+2)^2}{a^2(ck+1)^2}(1+\frac{\ln (1+ck)}{c})^2<\infty,
$$
and
$$
\sum_{k=1}^{\infty}Var[\tau^k(n)]<\sum_{k=1}^{\infty}\frac{(2c+2)^2}{a^2(ck+1)^2}\sum_{i=0}^{\infty}\frac{1}{(1+ic)^2}<\infty,
$$
we have
$$
\sum_{k=1}^{\infty} E[\tau^k(n)^2]<\infty.
$$
Using the same argument in the proof of Proposition 2,
we can deduce that
$\tau^k(n) \rightarrow 0$ almost surely when $k \rightarrow \infty$.
\end{proof}
}

{The above two propositions give us some insight about the infectious contagion,
when the portfolio size becomes large, 
with entities inside the portfolio being associated by infection, the contagion becomes more intensive.}

\subsection{Stochastic Intensity}

In this subsection, we extend the homogeneous contagion intensity process to the case that the constant intensity rate $a$ is replaced by an ``exogenous'' stochastic process $X$, i.e.,
\begin{equation}\label{p3}
\lambda_i(t) = X(t)\left(1+ \sum_{j=1,j \neq i}^{n} c 1_{\{ \tau_j \leq t\}} \right).
\end{equation}

\begin{proposition}
Suppose there are $n$ entities in our portfolio,
where the contagion stochastic intensity process follows (\ref{p3}).
Then the unconditional density functions of $\tau^k$, $k = 1, 2, \ldots, n$,
given the realization of $(X(s))_{0 \leq s<\infty}$, satisfy
the following recursive formula:
\begin{equation}\label{p4}
\begin{array}{l}
f_{\tau^{k+1}\mid (X(s))_{0 \leq s < \infty}}(t)
= (n-k)(1+kc) X(t) \int_0^t f_{\tau^k \mid (X(s))_{0 \leq s < \infty}}(u)e^{-(n-k)(1+kc)\int_u^{t} X(s)ds}du,
\end{array}
\end{equation}
for $k=1,2,\ldots,n-1$, where
$$
f_{\tau^1 \mid (X(s))_{0 \leq s < \infty}}(t)= nX(t)e^{-n \int_0^t X(u)du}.
$$
\end{proposition}

\begin{proof}
We note that in this case,
$\tau^{k+1}-\tau^k$ depends on the $k$th default time $\tau^k$,
in the way that the $(k+1)$th default rate $\lambda^{k+1}(t)$ follows:
$$
\lambda^{k+1}(t) = X(t)(1+kc)(n-k) \quad {\rm for} \quad \tau^k < t \leq \tau^{k+1}.
$$
Then we have the following key relationship
$$
\begin{array}{lll}
P(\tau^{k+1}-\tau^k>t \mid \tau^k, (X(s))_{0 \leq s < \infty})
&=& e^{-\int_0^t \lambda^{k+1}(\tau^k+s) ds}\\
&=& e^{-(1+kc)(n-k)\int_0^t X(\tau^k+s)ds}.
\end{array}
$$
Therefore, the density function of $\tau^{k+1}-\tau^{k}$
given the realization of $\tau^k$ and $(X(s))_{0 \leq s < \infty}$ is
$$
\begin{array}{lll}
f_{\tau^{k+1}-\tau^k \mid \tau^k, (X(s))_{0 \leq s < \infty}}(t)
&=& \displaystyle -\frac{d P(\tau^{k+1}-\tau^k>t \mid \tau^k, (X(s))_{0 \leq s< \infty})}{dt}\\
&=& \displaystyle (n-k)(1+kc)X(\tau^k+t)e^{-(n-k)(1+kc)\int_0^t X(\tau^k +s)ds}.
\end{array}
$$
Thus the density function of $\tau^{k+1}$
given the realization of $(X(s))_{0 \leq s<\infty}$, is given by
\begin{equation}\label{p4}
\begin{array}{lllll}
&&f_{\tau^{k+1}\mid (X(s))_{0 \leq s < \infty}}(t)\\
&=& \int_0^t f_{\tau^k \mid (X(s))_{0 \leq s < \infty}}(u)(n-k)(1+kc) X(t) e^{-(n-k)(1+kc)\int_0^{t-u} X(u+s)ds} du\\
&=& (n-k)(1+kc) X(t)\int_0^t f_{\tau^k \mid (X(s))_{0 \leq s < \infty}}(u)e^{-(n-k)(1+kc)\int_u^{t} X(s)ds}du,
\end{array}
\end{equation}
for $k=1,2,\ldots,n-1$, where
$$
f_{\tau^1 \mid (X(s))_{0 \leq s < \infty}}(t)
=f_{\tau^1-\tau^0 \mid (X(s))_{0 \leq s < \infty} }(t)=nX(t)e^{-n \int_0^t X(u)du}.
$$
\end{proof}

Again, applying the recursive formula (\ref{p4}) iteratively
gives the following corollary, which was
obtained in Zheng \& Jiang (2009).

\begin{corollary}
Assume that $c \neq 1/i$ for $i=1,2,\ldots,n-1$. Then the density function of $\tau^{k}$,
given $\mathcal{F}_{\infty}^X$, is given by
\begin{equation}
f_{\tau^k \mid (X(s))_{0 \leq s < \infty}}(t) = \sum_{j=0}^{k-1}\alpha_{k,j} X(t) e^{-\beta_j \int_0^t X(u)du},
\end{equation}
where $\alpha_{k,j}$ and $\beta_{j}$ are given in Corollary 1.
\end{corollary}

\section{Multi-state Stochastic Intensity Process}

In this section, we consider the stochastic intensity process to be a multi-state Markov process.
Here the state space $S=\{1,2,\ldots,N\}$
represents the set of all exogenous states.
For simplicity, we reduce the number of states to two.
Then the stochastic intensity process $X(t)$
alternates between $x_1$ and $x_2$, so that
$$
X(t)=\left\{
\begin{array}{ll}
x_1, \quad {\rm when \ the\  exogenous\  state \ lies\ in \ 1}\\
x_2, \quad {\rm when \ the \ exogenous \ state \ lies\  in \ 2}.
\end{array}\right.
$$
Let $\eta_i$ denote the rate of leaving state $i$ and $\pi_i$
the random time to leave state $i$, where $\pi_i$ is an exponential random variable, i.e.,
$$
P(\pi_i > t) = e^{-\eta_i t}, \ \ i=1,2.
$$
In this case, we consider a two-state, Markovian regime-switching,
intensity-based model for portfolio default risk.
Guo (1999) considers a two-state, Markovian regime-switching,
model for option valuation.
Here we follow her idea to derive the unconditional density function of $\tau^k$.
Let $T_i(t)$ be the total time between $0$ and $t$ during which
$X(s)=x_1$, starting from $X(0)=x_i$.
We then draw exponential random variables $\xi_1$ with intensity $\eta_1$ and $\xi_2$ with intensity $\eta_2$ independent of $X$.
Therefore
\begin{equation}\label{e2}
T_1(t)\hat{=} \left\{
\begin{array}{cll}
\xi_1 + T_2(t-\xi_1),& \xi_1 \leq t\\
t,& \xi_1 > t
\end{array}\right.,
\quad
T_2(t)\hat{=} \left\{
\begin{array}{cll}
T_1(t-\xi_2),&\xi_2 \leq t\\
0,& \xi_2 > t
\end{array}\right.,
\end{equation}
where ``$\hat{=}$'' means ``equals in distribution''.
Let
$$
\psi_i(l,t) =  E(e^{-l T_i(t)}) ,\  i=1,2,\  l \in \mathcal{R}.
$$

By Equation (\ref{e2}), we have
$$
\left\{
\begin{array}{l}
\psi_1(l,t) = \int_0^t \eta_1 e^{-(\eta_1+l)u} \psi_2 (l,t-u) du + e^{-(\eta_1+l)t},\\
\psi_2(l,t) = \int_0^t \eta_2 e^{-\eta_2 u} \psi_1 (l,t-u) du + e^{-\eta_2 t}.
\end{array}
\right.
$$
Taking the Laplace transform on both sides gives:
$$
\left\{
\begin{array}{l}
\mathcal{L}(\psi_1(l,t))(\cdot,s)=\frac{1}{l+s+\eta_1} + \frac{\eta_1}{l+s+\eta_1}\mathcal{L}(\psi_2(l,t))(\cdot,s),\\
\mathcal{L}(\psi_2(l,t))(\cdot,s)=\frac{1}{s+\eta_2} + \frac{\eta_2}{s+\eta_2}\mathcal{L}(\psi_1(l,t))(\cdot,s).
\end{array}
\right.
$$
Therefore,
$$
\mathcal{L}(\psi_1(l,t))(\cdot,s) = \frac{s+\eta_1+\eta_2}{s^2+ s\eta_1 +s\eta_2+ s l + l\eta_2}.
$$
By taking the inverse Laplace transform of the above equation, we have
$$
\psi_1(l,t) = \left\{
\begin{array}{ll}
e^{-\alpha t}[ \cos(\sqrt{\omega} t) + \frac{\beta}{\sqrt{\omega}} \sin(\sqrt{\omega} t)],& \omega > 0\\
e^{-\alpha t}(1+\beta t), & \omega=0\\
e^{-\alpha t}[ \cosh(\sqrt{-\omega} t) + \frac{\beta}{\sqrt{-\omega}} \sinh(\sqrt{-\omega} t)],&\omega<0,
\end{array}\right.
$$
where
$$
\alpha = \frac{\eta_1+\eta_2 +l}{2},\  \beta=\frac{\eta_1+\eta_2 -l}{2}, \ \omega = l\eta_2-\alpha^2.
$$
Similarly we have
$$
\mathcal{L}(\psi_2(l,t))(\cdot,s) = \frac{s+\eta_1+\eta_2 +l}{s^2+ s\eta_1 +s\eta_2+s l+l\eta_2},
$$
and
$$
\psi_2(l,t) = \left\{
\begin{array}{ll}
e^{-\alpha t}[ \cos(\sqrt{\omega} t) + \frac{\alpha}{\sqrt{\omega}} \sin(\sqrt{\omega} t)],& \omega > 0\\
e^{-\alpha t}(1+\alpha t), & \omega=0\\
e^{-\alpha t}[ \cosh(\sqrt{-\omega} t) + \frac{\alpha}{\sqrt{-\omega}} \sinh(\sqrt{-\omega} t)],& \omega<0.
\end{array}\right.
$$
We proceed with the derivation of the unconditional density function of $\tau^k$.
Given $X(0)=x_i$,
$$
E(e^{-L\int_0^t X(s) ds})=E(e^{-L(x_1T_i(t)+x_2(t-T_i(t)))})
=e^{-L x_2t}{\psi_i(L(x_1-x_2),t)} ,\ L>0\ .
$$
Then
$$
E(X(t)e^{-L\int_0^t X(s) ds})=-\frac{1}{L}{d(e^{-L x_2t}\psi_i(L(x_1-x_2),t)) \over dt}.
$$
Combining the results in Corollary 2,
the unconditional density function of $\tau^k$ can be obtained: when $X(0)=x_i$,
\begin{equation}
f_{\tau^k}(t) = -\sum_{j=0}^{k-1}\frac{\alpha_{k,j}}{\beta_j} {d(e^{-\beta_j x_2t}\psi_i(\beta_j(x_1-x_2),t))\over dt}.
\end{equation}

\section{The Heterogeneous Case}

In this section, we present our method in obtaining the
unconditional distribution of the $k$th default time in the heterogeneous case of
the interacting intensity default process.
For simplicity of discussion, we consider a two-group case.
The first group($G_1$) consists $n_1$ obligors and $\lambda_i(t)$ denotes the default rate of name $i$ in $G_1$
at time $t$, while the second group($G_2$) consists $n_2$ obligors and $\tilde{\lambda}_i(t)$ denotes the default rate of name $i$ in $G_2$ at time $t$. The interacting intensity process of the two-group case is assumed as follows:
\begin{equation}\label{p5}
\left\{
\begin{array}{lll}
\displaystyle \lambda_i(t)=a\left(1+b\sum_{j \neq i} 1_{\{\tau_j \leq t\}}+c \sum_{j} 1_{\{\tilde{\tau}_j \leq t\}}\right),\\
\displaystyle \tilde{\lambda}_i(t) = \tilde{a}\left(1+\tilde{b}\sum_j 1_{\{\tau_j \leq t\}}+\tilde{c} \sum_{j \neq i} 1_{\{\tilde{\tau}_j \leq t\}}\right),
 \end{array}
\right.
\end{equation}
where $\tau_j$ and $\tilde{\tau}_j$ denote the default time of name $j$ in $G_1$ and $G_2$, respectively, $a,\tilde a$ are positive constants and $b,c,\tilde b, \tilde c$ are nonnegative constants.
Let $N^i$ be the number of defaults in $G_1$ right after the $i$th default of our portfolio, where $N^0$ is assigned to be $0$.

\begin{proposition}
Suppose our portfolio has two groups of entities $G_1$ and $G_2$,
where $G_1$ and $G_2$ consist of $n_1$ and $n_2$ obligors, respectively, and
$n = n_1 + n_2$. For each $t \ge 0$, let $\lambda_i(t)$ and $\tilde{\lambda}_i(t)$
denote the default rates of name $i$ in
$G_1$ and $G_2$ at time $t$, respectively. These default rates
follow (\ref {p5}). For each $i = 1, 2, \cdots, n$, let
$N^i$ denote the number of defaults in $G_1$ right after
the $i$th default of our portfolio.
Then the recursive formula of the joint distribution of $\tau^k$ and $N^k$ is given by:
\begin{equation}\label{p6}
\begin{array}{lll}
f_{\tau^{k+1}, N^{k+1}}(t,m+1)
&=&\displaystyle {\tilde{\zeta}_{k,m+1}}\int_0^t f_{\tau^k,N^k}(u,m+1) e^{- (\zeta_{k,m+1}+\tilde{\zeta}_{k,m+1})(t-u)}du\\
& &\displaystyle + {\zeta}_{k,m} \int_0^t f_{\tau^k,N^k}(u,m)e^{- ({\zeta_{k,m}+\tilde{\zeta}_{k,m}})(t-u)}du,
\end{array}
\end{equation}
where
$$
\left\{
\begin{array}{l}
f_{\tau^{k},N^k}(t,m) = \displaystyle  -\frac{d P(\tau^k >t, N^k=m)}{d t},\\
\zeta_{k,m}=a(n_1-m)[1+bm+c(k-m)],\\
\tilde{\zeta}_{k,m}=\tilde{a}(n_2-(k-m))[1+\tilde{b}m+\tilde{c}(k-m)],
\end{array}
\right.
$$
and
$$
\left\{
\begin{array}{l}
f_{\tau^1,N^1}(t,1)=n_1 a e^{-(n_1a+n_2\tilde{a})t}, \\
f_{\tau^1,N^1}(t,0)=n_2 \tilde{a} e^{-(n_1a+n_2\tilde{a})t},\\
f_{\tau^1,N^1}(t,m)=0, \quad m \neq 0,1.
\end{array}
\right.
$$
\end{proposition}
\begin{proof}
Note that both $\tau^{k+1}-\tau^k$ and $N^{k+1}$ are independent of $\tau^k$, but
depend on $N^k$. Conditioning on $N^k$, $\tau^{k+1}-\tau^k$ and $N^{k+1}$ are independent.
The $(k+1)$th default rate $\lambda^{k+1}(t)$ that triggers $\tau^{k+1}$, given $N^k$, is given by
$$
\lambda^{k+1}(t)= \zeta_{k,N^k}+\tilde{\zeta}_{k, N^k} \quad {\rm for} \quad \tau^k < t \leq \tau^{k+1}.
$$
Hence, we have the following relations:
$$
P(\tau^{k+1}-\tau^k > t \mid N^k) = e^{-(\zeta_{k, N^k}+\tilde{\zeta}_{k,N^k})t}
$$
and
$$
f_{\tau^{k+1}-\tau^k \mid N^{k}}(t) = (\zeta_{k, N^k}+\tilde{\zeta}_{k,N^k}) e^{-(\zeta_{k, N^k}+\tilde{\zeta}_{k,N^k})t}.
$$
On the other hand,
$$
\begin{array}{lll}
P(N^{k+1}-N^{k}=1 \mid N^{k}) &=& P({\rm the} \ (k+1){\rm th \  default  \ happens \ in \ }G_1  \mid N^k )\\
&=& \displaystyle \frac{\zeta_{k,N^k}}{\zeta_{k,N^k}+\tilde{\zeta}_{k,N^k}},
\end{array}
$$
and similarly,
$$
P(N^{k+1}-N^{k}=0 \mid N^{k})= \frac{\tilde{\zeta}_{k,N^k}}{\zeta_{k,N^k}+\tilde{\zeta}_{k,N^k}},
$$
then the result follows.
\end{proof}

\begin{corollary}
Assume that $\beta_{k,m} \neq \beta_{i,j}$ for $1 \leq k \leq n, \max\{0,k-n_2\} \leq m \leq \min\{k,n_1\}$ and $i=0,\ldots, k-1, j=0,\ldots,m $. Then the joint density function of $\tau^{k}$ and $N^k$ is given by:
\begin{equation}
f_{\tau^k, N^k}(t,m) = \sum_{i=0}^{k-1} \sum_{j=0}^{m}\alpha_{k,m,i,j} e^{-\beta_{i,j} t},
\end{equation}
where the coefficients are given by the recursive formula:
$$
\left\{
\begin{array}{l}
\alpha_{k+1,m+1,i,j}= \left\{
\begin{array}{lll}
\displaystyle \frac{\alpha_{k,m+1,i,j} \tilde{\zeta}_{k,m+1}}{\beta_{k,m+1}-\beta_{i,j}} +\frac{\alpha_{k,m,i,j} \zeta_{k,m}}{\beta_{k,m}-\beta_{i,j}},& i=0,1,\ldots,k-1, j=0,\ldots,m\\
\displaystyle \frac{\alpha_{k,m+1,i,j} \tilde{\zeta}_{k,m+1}}{\beta_{k,m+1}-\beta_{i,j}}, &i=0,1,\ldots,k-1, j=m+1\\
\displaystyle -\sum_{u=0}^{k-1}\sum_{v=0}^{m+1}\frac{\alpha_{k,m+1,u,v} \tilde{\zeta}_{k,m+1}}{\beta_{k,m+1}-\beta_{u,v}},& i=k, j=m+1\\
\displaystyle -\sum_{u=0}^{k-1}\sum_{v=0}^{m}\frac{\alpha_{k,m,u,v} \zeta_{k,m}}{\beta_{k,m}-\beta_{u,v}},& i=k, j=m\\
0, &{\rm otherwise}
\end{array}
\right.\\
\displaystyle \beta_{i,j} = \zeta_{i,j}+\tilde{\zeta}_{i,j}
\end{array}
\right.
$$
and the boundry conditions are as follows:
$$
\alpha_{1,0,0,0}=n_2 \tilde{a},\  \alpha_{1,1,0,0}=n_1a,\  \alpha_{1,1,0,1}=0,\  \alpha_{1,m,i,j}=0,m \neq 0,1,
$$
and
$$
\alpha_{k+1,0,i,0}=
\left\{
\begin{array}{ll}
\displaystyle \frac{\tilde{\zeta}_{k,0}\alpha_{k,0,i,0}}{\beta_{k,0}-\beta_{i,0}},& i=0.\ldots,k-1\\
\displaystyle -\sum_{u=0}^{k-1}\frac{\tilde{\zeta}_{k,0}\alpha_{k,0,u,0}}{\beta_{k,0}-\beta_{u,0}},&i=k.
\end{array}
\right.
$$
\end{corollary}

As a result, the unconditional density function of $\tau^k$ is given by
$$
f_{\tau^k}(t) = \sum_{m=\max\{0,k-n_2\}}^{\min\{k,n_1\}}f_{\tau^{k}, N^{k}}(t,m).
$$

\section{Evaluation of Basket CDS and Sensitivity Study}

Consider a $k$th to default basket CDS with maturity $T$.
Assume $S_k$ is the $k$th swap rate, and $R$
is the recovery rate and $r$ is the annualized riskless interest rate.
The protection buyer A pays a periodic fee $S_k \Delta_i$
to the protection seller $B$ at time $t_i$, $i=1,2,\ldots,N$,
where $0=t_0<t_1<\ldots<t_N=T$ and $\Delta_i=t_i-t_{i-1}$.
If the $k$th default happens in the interval $[t_j,t_{j+1}]$,
$A$ will also pay $B$ the accrued default premium up to $\tau^k$.
On the other hand, if $\tau^k \leq T$, $B$ will pay $A$ the loss occurred at $\tau^k$,
that is, $1-R$.
Then the swap rate $S_k$ is given by
\begin{equation}{\label S}
S_k=\frac{(1-R)E(e^{-r\tau^k} 1_{\{\tau^k \leq T\}})}{\sum_{i=1}^{N}E(\Delta_i e^{-rt_i} 1_{\{\tau^k > t_i\}}+(\tau^k-t_{i-1})e^{-r\tau^k} 1_{\{t_{i-1} <\tau^k \leq t_i \}})}.
\end{equation}
{
We remark that one can reproduce the results in Zheng \& Jiang (2007) 
with our proposed methods  in a reasonably high accuracy.}

In Table 1, we present the swap rates under the intensity-based default contagion model with exponential decay. We compute the swap rates with $n=2, k=2$, by selecting different values of the parameters. One can see that when we fix $a, d$, the swap rate decreases while $c$ decreases. And when $a, c$ are fixed, the swap rate falls down while $d$ goes up. When $c, d$ are fixed, the swap rate increases when $a$ does.

Table 2 presents the swap rates under the two-state, Markovian regime-switching,
intensity-based default contagion model. Condition 1 is equivalent to the homogeneous case as $x_1=x_2$. And we can see that if $x_1<x_2$, swap rates increase when $\eta_1$ increases, while swap rates decrease when $\eta_2$ increases.

We present the swap rates under the heterogeneous intensity-based default contagion model in Table 3, where there are
two groups of entities and the intensity process follows (\ref{p5}). We conduct our numerical experiments via selecting different value of the parameters. As shown in Table 3, Conditions 1, 3 are equivalent to the homogeneous case for $b=\tilde{b}=c=\tilde{c}$. Condition 2 means that defaults in each group have a significant
impact on the entities of the same group, but have a marginal impact on the other group. Condition 4 shows the infection caused by entities in $G_1$ is much more intense than those in $G_2$. We also make comparison on the analytic pricing approach (AP) presented previously and Monte Carlo approach (MC) in this case. The analytic approach takes less than one second to compute all the swap rates of one column with MATLAB on a computer with an Intel 3.2 GHz CPU, while the Monte Carlo approach takes more than 5 minutes to run 100,000 simulations.

To study the sensitivities of the swap rates to a change in underlying parameters,
we presents the derivatives representing the sensitivities in the homogeneous case with the intensity (\ref{inten1}).
Note that
$$
\frac{\partial f_{\tau^k}(t)}{\partial a}=\sum_{j=0}^{k-1}\alpha_{k,j}(1-\beta_j at)e^{-\beta_j at},
$$
and
$$
\frac{\partial f_{\tau^k}(t)}{\partial c}=\sum_{j=0}^{k-1}(\alpha^{'}_{k,j} -\alpha_{k,j}(n-j)jat) ae^{-\beta_j at},
$$
where
$$
\left\{
\begin{array}{l}
\alpha^{'}_{k+1,j}= \left\{
\begin{array}{lll}
\displaystyle \alpha_{k,j} \gamma^{'}_{k,j}+ \alpha^{'}_{k,j} \gamma_{k,j},& j=0,1,\ldots,k-1\\
\displaystyle -\sum_{u=0}^{k-1}\alpha_{k,u} \gamma^{'}_{k,u}+ \alpha^{'}_{k,u} \gamma_{k,u},& j=k
\end{array}
\right.\\
\displaystyle \gamma_{k,j} = \frac{\beta_k}{\beta_k-\beta_j},\\
\displaystyle \gamma^{'}_{k,j} = \frac{\partial \gamma_{k,j}}{\partial c},
\end{array}
\right.
$$
and $\alpha^{'}_{1,0}=0$, $\alpha_{k,j}$ and $\beta_{j}$ are given in Corollary 1.
Then combining (\ref{S}), One can have
$$
\displaystyle \theta_{k}(a)=\frac{\partial S_k}{\partial a}
\quad {\rm and}  \quad \theta_{k}(c)=\frac{\partial S_k}{\partial c}.
$$
We present the derivatives $\theta_{k}(a)$,
$\theta_{k}(c)$ by using the recursive formula(AP) above (Figure 1, 2 (Left)), while we also use Monte Carlo method (MC) and difference quotient to find derivatives (Figure 1, 2 (Right)), { where we select step size as 0.1 in MC for difference quotient and 100,000 simulations have been done}.
We remark that using the Monte Carlo method to compute the derivatives is very time consuming,  and the results are unsatisfying, i.e., $\theta_1(c)$ is constant zero by definition,
while the results are quite fluctuating by Monte Carlo.

\begin{figure}
  \includegraphics[width=8cm]{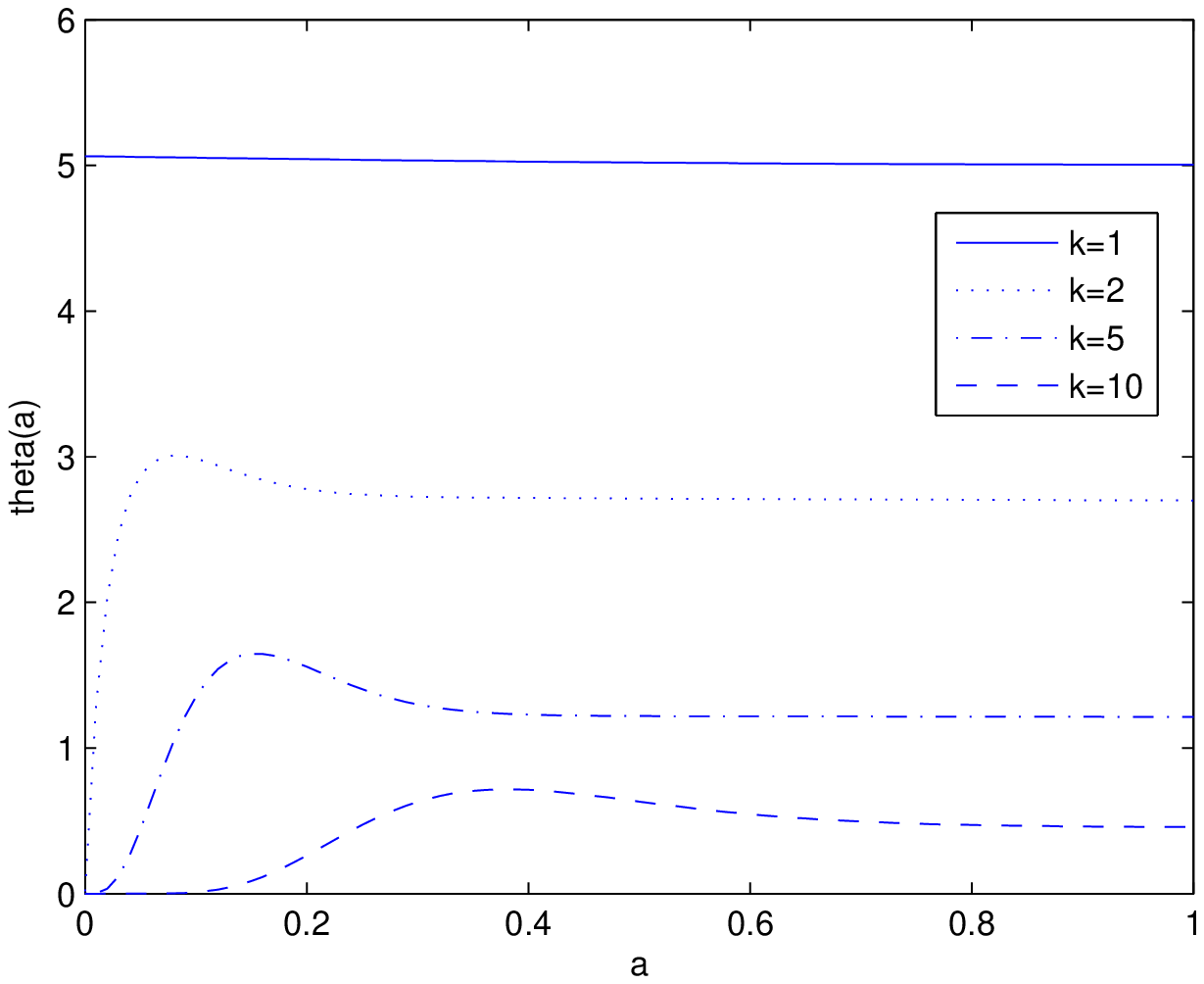} \includegraphics[width=8cm]{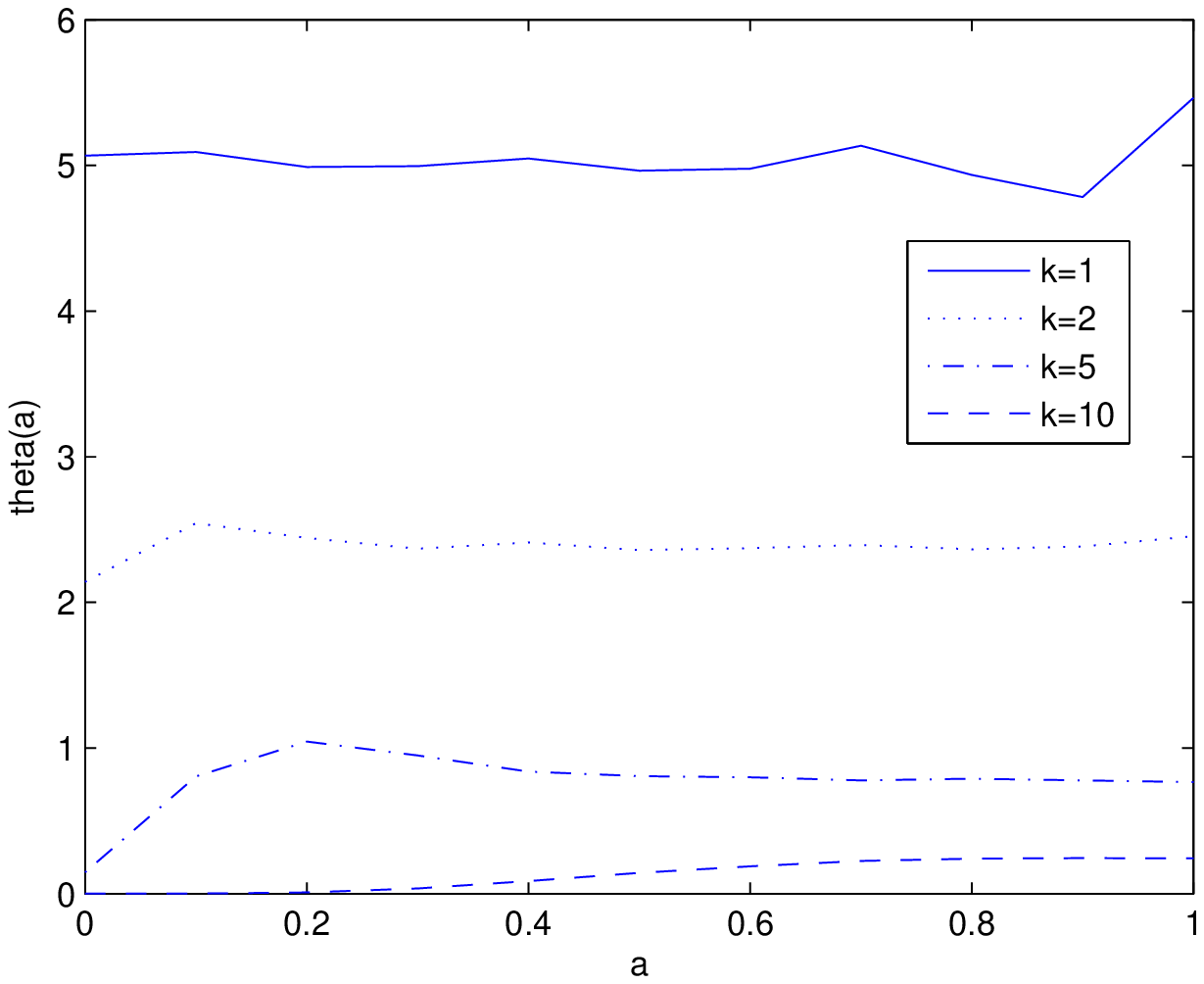}\\
  \caption{The derivatives of swap rates $S_k$ with respect to $a$ [left:AP, right:MC] ($n$=10, $c$=0.3)}\label{}
\end{figure}

\begin{figure}
  \includegraphics[width=8cm]{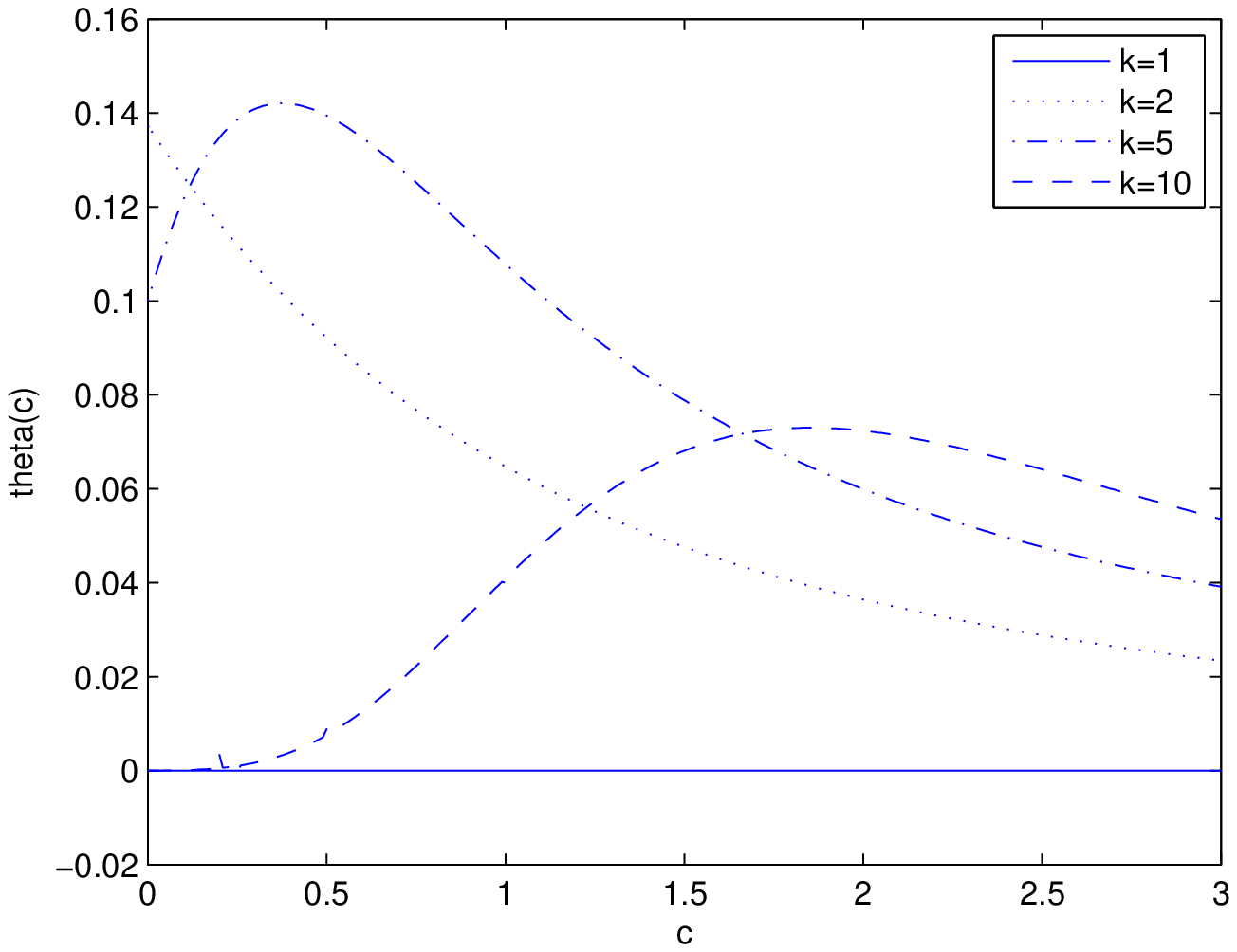} \includegraphics[width=8cm]{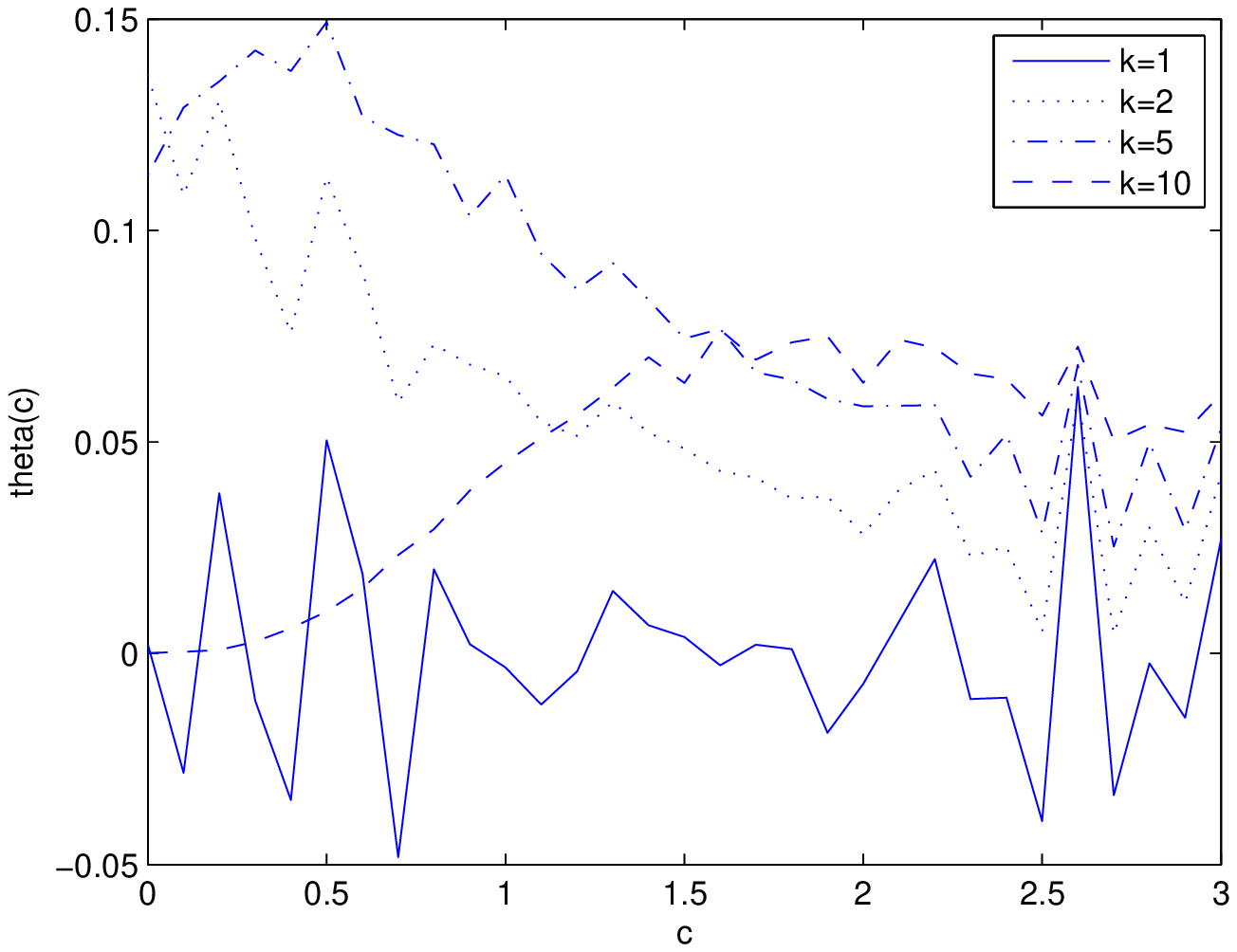}\\
  \caption{The derivatives of swap rates $S_k$ with respect to $c$ [left:AP, right:MC] ($n$=10, $a$=0.1)}\label{}
\end{figure}

\begin{table}\label{table1}
\caption{Basket CDS rates with exponential decay intensity-based default model
($n=2, k=2,  T=3, \Delta=0.5, R=0.5, r=0.05$) }
        \centering
                \begin{tabular}{|c|c|ccc|}
     \hline
                &$c$&0.2&1&5\\
\hline
                $a$&$d$ & & & \\
                \hline
        0.1&0.001&0.0134&0.0211&0.0479\\
           &0.01 &0.0134&0.0210&0.0477\\
           &0.1  &0.0132&0.0203&0.0459\\
           &1    &0.0123&0.0160&0.0322\\
           &10   &0.0115&0.0120&0.0147\\
           &100  &0.0114&0.0114&0.0117\\

                \hline
                1&0.001&0.3654&0.4961&0.7529\\
                 &0.01 &0.3651&0.4955&0.7526\\
                 &0.1  &0.3626&0.4898&0.7502\\
                 &1    &0.3464&0.4390&0.7184\\
                 &10   &0.3262&0.3447&0.4392\\
                 &100  &0.3222&0.3242&0.3342\\
                 \hline
                \end{tabular}
        \label{tab:1}
\end{table}

\begin{table}
\centering
\caption{Basket CDS rates in two-state stochastic intensity process case ($n=10, T=3, \Delta=0.5, R=0.5, r=0.05, c=3, X(0)=x_1$,
Condition 1: $x_1=1, x_2=1, \eta_1=1, \eta_2=1$,
Condition 2: $x_1=1, x_2=2, \eta_1=1, \eta_2=1$,
Condition 3: $x_1=1, x_2=2, \eta_1=1, \eta_2=2$,
Condition 4: $x_1=1, x_2=2, \eta_1=2, \eta_2=1$) }

\begin{tabular}{|c|c|c|c|c|}
  \hline
 $k$ & Condition 1 & Condition 2 & Condition 3 & Condition 4 \\
 \hline
 1 & 5.0242 & 5.2507 & 5.2409 & 5.4575   \\
 2 & 3.9288 & 4.1170 & 4.1087 & 4.2891   \\
 3 & 3.4456 & 3.6184 & 3.6106 & 3.7766   \\
 4 & 3.1369 & 3.3005 & 3.2930 & 3.4503   \\
 5 & 2.9035 & 3.0605 & 3.0532 & 3.2043   \\
 6 & 2.7070 & 2.8588 & 2.8516 & 2.9979  \\
 7 & 2.5270 & 2.6743 & 2.6672 & 2.8093   \\
 8 & 2.3473 & 2.4904 & 2.4833 & 2.6214   \\
 9 & 2.1459 & 2.2847 & 2.2775 & 2.4114  \\
10 & 1.8608 & 1.9945 & 1.9870 & 2.1159   \\
  \hline
\end{tabular}
\end{table}

\begin{table}
\centering
\caption{Basket CDS rates in heterogeneous case with 2 groups
($n_1=5, n_2=5, T=3, \Delta=0.5, R=0.5, r=0.05$,
Condition 1: $a=\tilde{a}=1, b=\tilde{b}=c=\tilde{c}=3$,
Condition 2: $a=\tilde{a}=1, b=\tilde{c}=3, \tilde{b}=c=0.3$,
Condition 3: $a=\tilde{a}=1, b=\tilde{b}=c=\tilde{c}=0.3$,
Condition 4: $a=\tilde{a}=1, b=\tilde{b}=3, c=\tilde{c}=0.3$)}
\begin{tabular}{|c|cc|cc|cc|cc|}
  \hline
 $k$ & Condition &1 & Condition &2 & Condition &3 & Condition &4 \\
 \hline
   &AP      & MC      &AP   & MC  &AP &  MC  &AP &MC  \\
 \hline
 1 & 5.0242 & 5.0265  & 5.0242 & 5.0352&5.0242 &5.0205 &5.0242 & 5.0463 \\
 2 & 3.9288 &  3.9352 & 3.4752 & 3.4692&2.7073 & 2.7167&3.2065  & 3.2167\\
 3 & 3.4456 &  3.4510 & 2.8287 & 2.8245&1.9036 & 1.9123&2.5866   &2.5922\\
 4 & 3.1369 & 3.1417  &2.4246 & 2.4209&1.4799 & 1.4860 &2.2543  & 2.2567\\
 5 & 2.9035 & 2.9062  &2.1161 &2.1135&1.2081 & 1.2095 &2.0302  & 2.0333\\
 6 & 2.7070 & 2.7068  &1.8376 & 1.8366&1.0112 & 1.0116 &1.8554  &1.8549\\
 7 & 2.5270 &  2.5270 &1.6445  & 1.6392&0.8550 & 0.8535 &1.7036  & 1.7013\\
 8 & 2.3473 & 2.3477 &1.4821  &1.4757&0.7203 & 0.7205 &1.5582  & 1.5545\\
 9 & 2.1459 & 2.1440 &1.3215 &1.3171&0.5921 & 0.5920 &1.4015  &1.3985\\
10 & 1.8608 & 1.8625 &1.1169  &1.1096&0.4451 & 0.4448 &1.1889  & 1.1851\\
  \hline
\end{tabular}
\end{table}

\section{Concluding Remarks}

In this paper we propose a simple recursive method
to compute the $k$th default time distribution under the interacting intensity default contagion model  (\ref{e1}).
We simplify the problem in the homogeneous case with exponential decay (\ref{p1})
and  with Markovian regime switching stochastic intensity (\ref{p3}). We further consider the problem in a two-group heterogeneous case  (\ref{p5}). 
We then present the numerical results for
the basket CDS rates and { sensitivity study} using the proposed method.
The main advantage of this method is that, by using the $k$th default rate, one can deduce the distributions of $k$th default times by recursive formulas with which one can easily compute CDS rates.
Moreover, the proposed method can also be applied to the heterogeneous case to obtain the analytic expressions of CDS rates.
Another key advantage is that, one can have the analytic formulas of the derivatives of the swap rates to the underlying parameters.
The analytic formulas are fast and accurate while the Monte Carlo method is slow and inaccurate as the numerical experiment reveals.\\

\noindent
{\bf Acknowledgment:} The authors would like to thank Prof. Mark H.A. Davis for
his helpful discussions and suggestions.

\end{document}